\documentclass{article}
\usepackage{amsmath}
\usepackage{amsthm}
\usepackage{amssymb}
\usepackage{stmaryrd}
\usepackage{caption}
\usepackage{enumerate}
\usepackage[lmargin=1in,rmargin=1in,tmargin=1in,bmargin=1in]{geometry}
\usepackage{algorithm}
\usepackage{algpseudocode}
\usepackage{blkarray}
\usepackage{tikz}
\usepackage{cleveref}
\usepackage{authblk}
\usetikzlibrary{automata,positioning}

\newcommand{\period}{{\bf Period}}
\algtext*{EndIf}
\algtext*{EndWhile}

\author[1,2]{Xiaoyu He\thanks{hexiaoyu14@mails.ucas.ac.cn}}
\author[2]{Neng Huang\thanks{huangneng14@mails.ucas.ac.cn}}
\author[1,2]{Xiaoming Sun\thanks{sunxiaoming@ict.ac.cn}}
\affil[1]{Institute of Computing Technology, Chinese Academy of Sciences, China}
\affil[2]{University of Chinese Academy of Sciences, China}

\theoremstyle{plain}
\newtheorem{theorem}{Theorem}[section]
\newtheorem{lemma}[theorem]{Lemma}
\newtheorem{corollary}[theorem]{Corollary}
\newtheorem{prop}[theorem]{Proposition}

\theoremstyle{definition}
\newtheorem{definition}[theorem]{Definition}
\theoremstyle{remark}
\newtheorem{remark}{Remark}

\begin{document}

\title{On the Decision Tree Complexity of String Matching\footnote{This work was supported in part by the National Natural Science Foundation of China Grant 61433014, 61502449, 61602440, and the 973 Program of China Grants No. 2016YFB1000201.}
}
\maketitle

\begin{abstract}
 String matching is one of the most fundamental problems in computer science. A natural problem is to determine the number of characters that need to be queried (i.e. the decision tree complexity) in a string in order to decide whether this string contains a certain pattern.
 Rivest showed that for every pattern $p$,  in the worst case any deterministic algorithm needs to query at least $n-|p|+1$ characters, where $n$ is the length of the string and $|p|$ is the length of the pattern. He further conjectured that this bound is tight.
 By using the adversary method, Tuza disproved this conjecture and showed that more than one half of binary patterns are {\em evasive}, i.e. any algorithm needs to query all the characters (see Section 1.1 for more details).
 
  In this paper, we give a query algorithm which settles the decision tree complexity of string matching except for a negligible fraction of patterns. Our algorithm shows that Tuza's criteria of evasive patterns are almost complete. Using the algebraic approach of Rivest and Vuillemin, we also give a new sufficient condition for the evasiveness of patterns, which is beyond Tuza's criteria. In addition, our result reveals an interesting connection to \emph{Skolem's Problem} in mathematics.

\end{abstract}

\section{Introduction}
The string matching problem is one of the most fundamental problems in computer science. The goal of string matching problem is to find one or all occurrences of a pattern in an input string. Lots of efficient algorithms have been discovered in the 20th century. For example, the KMP algorithm~\cite{knuth_fast_1977}, discovered by Knuth, Morris and Pratt, is able to locate all occurrences of a pattern of length $m$ in a string of length $n$ in $O(n + m)$ time. This is essentially the best possible since every algorithm needs $\Omega(n + m)$ time to process the input. Another elegant algorithm is the Karp-Rabin algorithm~\cite{karp_efficient_1987}, which uses hashing and can be used to search for a set of patterns. A detailed treatment of these algorithms can be found in \cite{cormen_introduction_2001}. However, the problem becomes subtler when we adopt a different complexity measure, which is the number of characters that the algorithm has to examine in the input string given the prior knowledge of the pattern string. When we confine the alphabet to $\{0, 1\}$, this measure is exactly the \emph{decision tree complexity} of boolean string matching problem. Recall that for a binary function $f$, its decision tree complexity is the number of bits that we have to examine in the worst case in any input $x$ in order to compute $f(x)$.

\subsection{Notations and Previous Work}\label{sec:sec_1.1}
Let $p$ be a pattern over alphabet $\Sigma$ with $|\Sigma| = \sigma$. Throughout the paper, let $p[i]$ be the $i$-th character in $p$ and $p[i..j]$ be the substring of $p$ indexed from $i$ to $j$. Let $A_p$ be a deterministic string searching algorithm which searches for $p$ in any given string $s$. Following Rivest~\cite{rivest_worst-case_1977}, we denote $w(A_p, n)$ to be the maximum number of characters that $A_p$ examines for any $s$ of length $n$. Let $D_p(n) = \min_{A_p} w(A_p, n)$, where $A_p$ is taken over all deterministic string searching algorithms. When the alphabet $\Sigma = \{0, 1\}$, $D_p$ is exactly the boolean decision tree complexity of string searching algorithm with pattern $p$. It is clear that this function is monotone, since we can simply add some redundant characters at the end of the searched text. We state this as the following proposition.

\begin{prop}\label{prop:L1}
For every pattern $p$ and $n \in \mathbb{N}$, $D_p(n) \leq D_p(n + 1)$.
\end{prop}

We define the evasiveness of a pattern as follows.
\begin{definition}
A pattern $p$ is called \emph{evasive} if there exists $N_0 \in \mathbb{N}$ such that for all $n > N_0$, $D_p(n) = n$.
\end{definition}

By this definition, a pattern $p$ is evasive if for every algorithm $A$ and every sufficiently large $n$, there is a string $s$ of length $n$ such that $A$ has to query every character of $s$ in order to determine whether $s$ contains $p$ as a substring. We are interested in determining what patterns are evasive and what patterns are not.

Let $|p|$ denote the length of a pattern $p$. Rivest gave the following linear lower bound on $D_p(n)$:
\begin{theorem}[\cite{rivest_worst-case_1977}]\label{thm:rivestthm}
For every pattern $p$, $D_p(n) \geq n - |p| + 1$ for all $n \in \mathbb{N}$.
\end{theorem}

To prove this theorem, Rivest showed that for every $n \in \mathbb{N}$, there exists an integer $i$ between 0 and $|p|$ such that $D_p(n + i) = n + i$, then combined with Proposition~\ref{prop:L1}, Theorem~\ref{thm:rivestthm} follows. Based on this result, we define {\em non-evasiveness} as follows.
\begin{definition}
A pattern $p$ is called \emph{non-evasive} if for every $N_0 \in \mathbb{N}$ there exists $n>N_0$ such that $D_p(n) = n - |p| + 1$.
\end{definition}
As discussed above, what Rivest proved in fact implies that it is impossible for a pattern to achieve the lower bound in Theorem~\ref{thm:rivestthm} on consecutive integers, which is the reason we define non-evasiveness in this way. Rivest showed that the pattern $p = 1^k$ (and therefore $0^k$) is non-evasive. He further conjectured that all patterns are non-evasive. However, this conjecture was later disproved by Tuza in~\cite{tuza_worst-case_1982}. We briefly summarize Tuza's work here. Given a string $b$, let $BE(b)$ denote the set of patterns prefixed and suffixed by $b$, but other than $b$. Also, for patterns $u$ and $v$, let $uv$ denote their concatenation. If $p \in BE(b)$, then let $p(b)$ be the string $ubv$ where, $ub = bv = p$. Tuza proved the following result.

\begin{theorem}[\cite{tuza_worst-case_1982}]\label{thm:tuzathm}
Let $p \in BE(b)$. If
\begin{enumerate}
  \item $p(b)$ does not contain a substring $p'$ of length $|p|$ other than prefix or suffix of $p(b)$ such that $p'$ and $p$ differ from each other in at most two characters, and
  \item the pattern $pp$ does not contain a substring $p'$ of length $|p|$ other than prefix or suffix of $pp$ such that $p'$ and $p$ differ from each other in at most four characters,

\end{enumerate}
 and $n \geq |p|(2|p| - |b|)/\gcd(|p|, |b|)$, then $D_p(n) \geq n - k$, where $k = n \mod \gcd(|p|, |b|)$.
\end{theorem}

If a pattern string $p$ satisfies the conditions in Theorem~\ref{thm:tuzathm} and $\gcd(|p|, |b|) = 1$, then one would have $D_p(n) = n$ for all sufficiently large $n$. This implies that $p$ is evasive and therefore serves as a counterexample of Rivest's conjecture. Tuza estimated the proportion of pattern strings which satisfy the conditions in Theorem~\ref{thm:tuzathm} and proved that when $\Sigma = \{0, 1\}$, there exists more than $0.5061\cdot 2^m$ evasive patterns of length $m$.

Beyond the worst case complexity, the average-case complexity has also been studied previously, that is, finding out the numbers of characters that need to be examined on average assuming that the input string is sampled from the uniform distribution. Yao~\cite{yao_complexity_1979} showed that, for almost all patterns of large enough length $m$, an algorithm needs to examine $\Theta(\frac{n \log_q m}{m})$ characters on a uniformly random input string of length $n > 2m$, here $q$ is the size of the alphabet.

\subsection{Our Contributions}
In this paper we settle the decision tree complexity for almost every string except an $o(1)$ fraction. More precisely, we prove that Tuza's lower bound, which is developed combinatorially, is in fact tight for almost every string, by showing an algorithm which achieves this lower bound. This algorithm is based on the periods of the pattern string.

\begin{definition}[Periods]
Let $p$ be a pattern of length $m$ and $k$ be a positive integer no larger than $m$. We say that $p$ is $k$-periodic, or $p$ has a period $k$, if $p[i] = p[i + k]$ for all $1 \leq i \leq n - k$. Let ${\bf Period}(p)=\{k|p\ is\ k\text{-periodic}\}$ be the set of all periods of $p$.
\end{definition}
The definition here is the same as in~\cite{galil_time-space-optimal_1983}, in which it was used to develop a time-space optimal algorithm. A similar idea can also be found in~\cite{tuza_worst-case_1982}. For set $S\subset \mathbb{N}$, let $\gcd(S)$ denote the greatest common divisor of all elements in $S$. Here is our main theorem. 

\begin{theorem}[Main]\label{thm:mainthm}
Let $p$ be a pattern of length $m$ and $c=\gcd({\bf Period}(p))$ be the greatest common divisor of all $p$'s periods, then $D_p(n) = n - (n \mod c)$, except for an $O(m^5\sigma^{-m/2})$ fraction of patterns.
\end{theorem}
Here, the fraction of patterns is computed in the following way. We first fix a pattern length $m$, and then count the number of patterns of length $m$ that satisfy some certain properties, then compute its ratio to the total number of length-$m$ patterns, which is $\sigma^m$. We then investigate the asymptotic behavior of this ratio as $m$ goes to infinity.

By Theorem~\ref{thm:mainthm}, the fraction of patterns whose decision tree complexity we don't know goes to 0 as the pattern length goes to infinity.

Besides this result, we also use the algebraic approach to show the evasiveness of certain family of patterns, for which Tuza's method does not work. This algebraic approach was first developed by Rivest and Vuillemin~\cite{rivest_generalization_1975}, and we extend it to our problem. Interestingly, we find that this approach reveals a relation between our problem and the \emph{Skolem's Problem}. We also define the \emph{characteristic polynomial} of a pattern, which is again closely related to the pattern's periodic behaviors. This polynomial, besides its application in this problem, is of independent interest on its own.

\section{Upper Bounds}
In this section, we prove one direction of Theorem~\ref{thm:mainthm}, which can be stated as the following lemma.
\begin{lemma}\label{lem:upperbd}
Let $p$ be a pattern of length $m$ and $c=\gcd({\bf Period}(p))$ be the greatest common divisor of all $p$'s periods, then $D_p(n) \leq n - (n \mod c)$.
\end{lemma}
To show this lemma, we will develop an algorithm whose behavior depends on the periods of the pattern string.
\subsection{Non-evasiveness of Bifix-free Patterns}
We first look at the simple case where our pattern is bifix-free.

\begin{definition}
A string $s$ is called a \emph{bifix} of a string $t$ if $s$ is both a prefix and a suffix of $t$. A pattern $p$ is called \emph{bifix-free} if $p$ has no bifix other than itself.
\end{definition}

\begin{remark}[Relations to combinatorics on words]
The concepts of periods and bifixes are also studied in the field of combinatorics on words under possibly different names. Bifixes are usually referred to as \emph{borders} in combinatorics on words, and bifix-free strings are usually called \emph{unbordered words}. For more details from viewpoint of combinatorics on words, see~\cite{berstel_combinatorics_2003}.
\end{remark}

Bifix-free patterns have the following property in terms of periods.
\begin{lemma}\label{lem:bifixprop}
A pattern $p$ of length $n$ has a bifix of length $k < n$ if and only if it is $(n - k)$-periodic. Furthermore, $p$ is bifix-free if and only if it has only one period, which is $n$.
\end{lemma}
\begin{proof}
If a pattern $p$ has a period $k < n$, then $p[1..(n-k)] = p[(k+1)..n]$. This is equivalent to say that $p$ has a prefix of length $n - k$ which is equal to $p$'s suffix of length $n - k$. The ``furthermore'' part follows directly.
\end{proof}

Then, for a bifix-free pattern $p$ we have $|p|=\gcd({\bf Period}(p))$. According to Lemma~\ref{lem:upperbd}, the following result is expected.

\begin{lemma}\label{lem:bifixthm}
Let $p$ be a bifix-free pattern of length m, then $D_p(n) \leq n - (n \mod m)$ and $p$ is non-evasive.
\end{lemma}
\begin{proof}
Consider the algorithm in Figure~\ref{fig:A1}. We claim that this algorithm can produce the correct output after $n - (n \mod m)$ queries to the string.
\begin{figure}[h]
\centering
\fbox{\parbox{0.8\linewidth}{
  \begin{algorithmic}[1]
    \Require string $s$ of length $n$, bifix-free pattern $p$ of length $m$
    \Ensure whether $p$ is a substring of $s$
    \Function{Find}{$s$, $p$}
    \If {$n < m$}
        \State \Return $false$
    \EndIf
    \State $i \gets m$, $j \gets m$    
    \State query($s[m]$)
    \While {$j - i \neq m - 1$}
        \If {$s[i..j]$ is a suffix of $p$}
            \State query($s[i-1]$), $i \gets i-1$
        \Else
            \State query($s[j+1]$), $j \gets j+1$
        \EndIf
    \EndWhile
    \If {$s[i..j] = p$}
        \State\Return $true$
    \Else
        \State\Return {\textsc{Find}($s[m+1..n]$, $p$)}
    \EndIf
    \EndFunction
  \end{algorithmic}
  }}
\caption{Algorithm for bifix-free patterns}
\label{fig:A1}
\end{figure}
Suppose that in Line 11, we find that $s[i..j]$ is not equal to $p$, otherwise we can stop and output this occurrence. Note that until Line 11 we have only queried $m$ characters in $s$, which are $s[i], s[i + 1], \ldots, s[j]$. We show that for indices $l$ with $1 \leq l \leq m$, we have $s[l..(l+m-1)] \neq p$.
\begin{itemize}
  \item $1 \leq l < i$. In this case, there exists an index $t$ with $i \leq t \leq l + m - 1$ such that $s[t..(l + m - 1)]$ is not a suffix of $p$, since otherwise $s[l + m]$ would not be queried, contradicting the fact that $j = i + m - 1 \geq l + m$. And therefore we have $s[l..(l+m-1)]$ is not suffix of $p$.
  \item $l = i$. In this case we have by assumption that $s[l..(l + m - 1)] = s[i..j] \neq p$.
  \item $i < l \leq m$. Assume that $s[l..(l+m-1)]$ equals to $p$. Then for all indices $t$ with $l \leq t \leq j$, $s[l..t]$ is a prefix of $p$, and therefore by bifix-freeness, is not a suffix of $p$. However, since $i < l$, $s[l-1]$ is queried, so there must exists such an index $t$ that $s[l..t]$ is a suffix of $p$, which is a contradiction. Hence $s[l..(l+m-1)]$ does not equal to $p$.
\end{itemize}
This shows that after querying $m$ characters, we either find an occurrence of $p$ in $s$, or reduce the size of $s$ by $m$. When the size of $s$ is smaller than $m$, the algorithm trivially stops. Therefore after $n - (n \mod m)$ queries, we will be able to determine whether $s$ contains $p$. This establishes an upper bound on $D_p(n)$, namely $D_p(n) \leq n - (n \mod m)$, which matches Rivest's lower bound. We conclude that bifix-free patterns are non-evasive.
\end{proof}

We note that the above algorithm is in fact applicable for all finite alphabets. For an alphabet $\Sigma$ of size $\sigma$, we define $b_m^\sigma$ to be the proportion of bifix-free strings in strings of length $m$, that is,
\[
b_m^\sigma = \frac{|\{p \in \Sigma^m | p \text{ is bifix-free}\}|}{\sigma^m}, |\Sigma| = \sigma.
\]

Nielsen~\cite{nielsen_note_1973} showed that the sequence $\{b_m^\sigma\}_{m=1}^\infty$ converges. Furthermore, he proved that
\[
b_\infty^\sigma := \lim_{m \rightarrow \infty}b_m^\sigma \geq 1 - \sigma^{-1} - \sigma^{-2}.
\]
The following table from \cite{nielsen_note_1973} shows the first three significant digits for $b_\infty^\sigma$ when $\sigma \leq 6$.
\begin{figure}[h]
\centering
\begin{tabular}{|c|c|c|c|c|c|}
  \hline

  $\sigma$ & 2 & 3 & 4 & 5 & 6 \\ \hline
  $b_\infty^\sigma$ & 0.268 & 0.557 & 0.688 & 0.760 & 0.801 \\
  \hline
\end{tabular}
\caption{The table for the first three significant digits for $b_\infty^\sigma$ when $\sigma \leq 6$}
\label{fig:F2}
\end{figure}

From this we obtain that more than 26.7\% of binary pattern strings of length $m$ are non-evasive, where $m$ is sufficiently large. We also note that, as the size of the alphabet increases, the percentage of patterns that are non-evasive tends to 1.

\subsection{The General Case}
In the previous section, we used bifix-freeness as a crucial tool in our algorithm. The property stated in Lemma~\ref{lem:bifixprop} is in fact playing an important role here. It is natural to ask that what if a pattern has periods other than its own length? An intuition is that if a pattern has good periodic behaviors, then a well-behaved algorithm must exist as well. We therefore formalize this intuition and give the proof of Lemma~\ref{lem:upperbd}.

\begin{proof}[Proof of Lemma~\ref{lem:upperbd}]

Let's consider the algorithm in Figure~\ref{fig:A2}, which is a generalization of the algorithm for bifix-free patterns. Intuitively, this algorithm examines the string by blocks of size $c$, which is the greatest common divisor of $p$'s periods. Note that for simplicity we formulate this algorithm in a way that it may query the same character more than once. In such cases, we can reuse the previous result and need not really query that character. Our algorithm might also query a character in $s$ with index larger than $n$. In such cases, we assume that we obtain a character different than $p[m]$, such that it cannot form the pattern $p$ with previous characters. We also assume that $c > 1$.

\begin{figure}[h]

\centering
\fbox{\parbox{0.8\linewidth}{
  \begin{algorithmic}[1]
    \Require string $s$ of length $n$, pattern $p$ of length $m$
    \Ensure whether $p$ is a substring of $s$
    \Function{Find}{$s$, $p$}
    \If {$n < m$}
        \ \Return $false$
    \EndIf
    \State $i \gets m$, $j \gets m$
    \State query($s[m]$)
    \While{$j - i \neq m - 1$}
        \If {$s[i..j]$ is a suffix of $p$}
            \State query($s[i-1]$), $i \gets i-1$
        \Else
            \State query($s[j+1]$), $j \gets j+1$
        \EndIf
    \EndWhile
    \If {$s[i..j] = p$}
        \State\Return $true$
    \EndIf
    \State $l \gets m+c$
    
    \While{$l \leq n$}
    
      \State $i \gets l$, $j \gets l$
      \State query($s[l]$)
      \Repeat
        \If {$s[i..j]$ is a suffix of $p$}
            \State query($s[i-1]$), $i \gets i-1$
        \Else
            \State query($s[j+1]$), $j \gets j+1$
        \EndIf
      \Until $c$ new characters have been queried \textbf{OR} $j - i = m - 1$
      
        \If {$s[(j-m+1)..j] = p$}
            \State\Return $true$
        \EndIf
      \State $l \gets l + c$
    \EndWhile

    \State\Return $false$
    \EndFunction
  \end{algorithmic}
  }}
\caption{Algorithm for general patterns}
\label{fig:A2}
\end{figure}

First of all, it is easy to see that this algorithm queries at most $n - (n \mod c)$ characters in $s$. We now show that this algorithm returns the answer correctly. Our algorithm only returns $true$ when it really see the pattern $p$, so it suffices to show that if there are occurrences of $p$ in $s$, then our algorithm will always be able to find one. Here we prove that it will always find the first occurrence.

Assume that the first occurrence of $p$ in $s$ is $s[k - m + 1..k]$ and $k = hc + t$ for some $0 \leq t < c$. We want to show that, when our algorithm starts to examine the $(hc)$-th character of the string at Line 15 (it could be that our algorithm will be able to locate $p$ in the while loop beginning at Line 5, but that case is even simpler), there are at most $c$ characters in $s[k - m + 1..k]$ which have not been queried yet. If this holds, then our algorithm will be able to identify $s[k-m+1..k]$ as $p$ in at most $c$ queries.

In fact, we prove a strong claim that whenever $k-m < l \leq k$, in order for the repeat-until loop at Line 16-21 to stop, we either either query $c$ new characters in the range $s[k-m+1..k]$, or we have queried every character in the range $s[k-m+1..l+t]$.

Suppose our algorithm is going to query a character with index smaller than $k-m+1$ when $k-m < l \leq k$, then at some point our algorithm will query $s[k - m]$ at Line 18. Clearly, $i = k - m + 1$ at that moment. Also, it must be that $j = l + t$, for when Line 18 is executed, $s[i..j]$ must be a suffix of $p$. But we also know that $s[i..j]$ is a prefix of $p$. Thus the length of $s[i..j]$, which is $j - k + m$, must be a multiple of $c$, implying that $j = l + t$ (since $l$ is always a multiple of $c$). If our algorithm do not query a character with index smaller than $k-m+1$ when $k-m < l \leq k$,  then in order for the repeat-until loop at Line 16 to end, we have all our $c$ new characters in $s[k-m+1..k]$. This proves what we need, and the correctness of our algorithm follows.
\end{proof}

\section{Proof of Theorem~\ref{thm:mainthm}}
In this section, we give the proof of our main theorem. We have proved one direction in Section 2. For the other direction, we use a similar analysis to Tuza's in~\cite{tuza_worst-case_1982}. We first restate (a stronger version of) Tuza's theorem here.
\begin{theorem}[\cite{tuza_worst-case_1982}]\label{thm:tuzathm_strengthened}
Assume that $p \in BE(b_1), p \in BE(b_2), \ldots, p \in BE(b_l)$. If
\begin{enumerate}
  \item for every $1 \leq i \leq l$, $p(b_i)$ does not contain a substring $p'$ of length $|p|$ other than prefix or suffix of $p(b_i)$ such that $p'$ and $p$ differ from each other in at most two characters, and
  \item the pattern $pp$ does not contain a substring $p'$ of length $|p|$ other than prefix or suffix of $pp$ such that $p'$ and $p$ differ from each other in at most four characters,

\end{enumerate}
then for sufficiently large $n$, $D_p(n) \geq n - k$, where $k = n \mod \gcd(\{|p|, |b_1|, |b_2|, \ldots, |b_l|\})$.
\end{theorem}

We note that in Tuza's language, $p \in BE(b)$ essentially means that $p$ has a bifix $b$. As is shown in Lemma~\ref{lem:bifixprop}, it is equivalent to say that $p$ is $(|p| - |b|)$-periodic. Thus the condition $p \in BE(b_1), p \in BE(b_2), \ldots, p \in BE(b_l)$ is simply saying that $p$ has periods $|p| - |b_1|, |p| - |b_2|, \ldots, |p| - |b_l|$, other than its own length $|p|$, and the expression $\gcd(\{|p|, |b_1|, |b_2|, \ldots, |b_l|\}$ is equivalent to $\gcd({\bf Period}(p))$. To prove Theorem~\ref{thm:mainthm}, we need the following two lemmas. These two lemmas are generalizations of Lemma 11 and Lemma 12 in~\cite{tuza_worst-case_1982}.

\begin{lemma}\label{lem:lem1}
Let $B_1(n)$ be the set of patterns $p$ such that $|p| = n$, $p \in BE(b)$ for some $b$ and $p(b)$ contains a substring $p'$ of length $n$ other than prefix or suffix of $p(b)$ such that $p'$ and $p$ differ from each other in at most two characters. Then $|B_1(n)| = O(n^4\sigma^{n/2})$.
\end{lemma}

\begin{lemma}\label{lem:lem2}
Let $B_2(n)$ be the set of patterns $p$ such that $|p| = n$ and the pattern $pp$ contains a substring $p'$ of length $n$ other than prefix or suffix of $pp$ such that $p'$ and $p$ differ from each other in at most four characters. Then $|B_2(n)| = O(n^5\sigma^{n/2})$.
\end{lemma}

\begin{proof}[Proof of Theorem~\ref{thm:mainthm}]
Let $p$ be a pattern of length $m$. If $p \notin B_1(m) \cup B_2(m)$, then by Theorem~\ref{thm:tuzathm_strengthened}, $D_p(n) \geq n - k$, where $k = n \mod \gcd(\period(p))$. Also, by Lemma~\ref{lem:upperbd}, $D_p(n) \leq n - k$. Therefore $D_p(n) = n - k$ for all $p \notin B_1(m) \cup B_2(m)$. By Lemma~\ref{lem:lem1} and Lemma~\ref{lem:lem2}, $|B_1(m) \cup B_2(m)| = O(m^5\sigma^{m/2})$, and hence Theorem~\ref{thm:mainthm} follows.
\end{proof}

For simplicity, from now on we say $p(b)$ \emph{has property 1} if $p \in BE(b)$ and $p(b)$ contains a substring $p'$ of length $n$ other than prefix or suffix of $p(b)$ such that $p'$ and $p$ differ from each other in at most two characters, and we say $p$ \emph{has property 2} if the pattern $pp$ contains a substring $p'$ of length $n$ other than prefix or suffix of $pp$ such that $p'$ and $p$ differ from each other in at most four characters.

\subsection{Proofs of the Two Lemmas}
Now we prove Lemma~\ref{lem:lem1} and Lemma~\ref{lem:lem2}. Tuza proved the case when $\Sigma = \{0, 1\}$ in~\cite{tuza_worst-case_1982}. We will adapt his proof to handle the case where $\Sigma$ is any finite alphabet.

\begin{lemma}\label{lem:prop1}
If $p(b)$ has property 1 for some $|b| > |p|/2$, then we can find $b'$ with length at most $|p| / 2$ such that $p(b')$ has property 1 as well.
\end{lemma}
\begin{proof}
If $p \in BE(b)$ for some $|b| > |p|/2$, then by definition, $p[i] = p[i + |p| - |b|]$ for every $1 \leq i \leq |b|$. Assume that $k(|p| - |b|) < |p| \leq (k + 1)(|p| - |b|)$ for some $k$, then let $b' = p[k(|p| - |b|) + 1..|p|]$. It is straightforward to check that $b'$ is also a bifix of $p$ and $p(b')$ contains $p(b)$ as a substring. Therefore $p(b')$ has property 1 if $p(b)$ has property 1.
\end{proof}

\begin{proof}[Proof of Lemma~\ref{lem:lem1}]
Let $p \in B_1(n)$. By definition and Lemma~\ref{lem:prop1}, for some $|b| \leq |p| / 2$, $p(b)$ contains a substring $p'$ of length $n$ other than prefix or suffix of $p(b)$ such that $p'$ and $p$ differ from each other in at most two characters. These at most two characters can be chosen in $(\sigma - 1)^2n(n - 1)/2 + (\sigma - 1)n + 1$ different ways. Assume that $p'$ starts in the $(i + 1)$-th character in $p(b)$, then after we fix these two erroneous locations, the first $\gcd(i, n - |b|)$ characters in $p(b)$ will uniquely determine $p(b)$. Therefore we have
\begin{align*}
|B_1(n)| & \leq \sum_{|b| = 1}^{n/2}\sum_{i=1}^{n - |b| - 1}((\sigma - 1)^2n(n - 1)/2 + (\sigma - 1)n + 1)\sigma^{\gcd(i, n - |b|)} \\
& \leq n^2\sigma^2\sum_{|b| = 1}^{n/2}\sum_{i=1}^{n - |b| - 1}\sigma^{\gcd(i, n - |b|)} \\
& \leq n^2\sigma^2\cdot\frac{n}{2}\cdot n\sigma^{n/2} \\
& = \frac{n^4}{2}\sigma^{n/2 + 2}.
\end{align*}
\end{proof}

\begin{proof}[Proof of Lemma~\ref{lem:lem2}]
The proof is similar to that of Lemma~\ref{lem:lem1}. Let $p \in B_2(n)$. Then the pattern $pp$ has a substring $p'$ that differs from $p$ in at most four characters. These at most four characters can be chosen in at most $|\Sigma|^4n^4$ different ways. Assume that $p'$ starts in the $(i + 1)$-th position, then the first $\gcd(i, n)$ characters in $p$ uniquely determines $p$. Therefore we have
\[
|B_2(n)| \leq \sum_{i=1}^{n - 1}\sigma^4n^4\sigma^{\gcd(i, n)} \leq \sum_{i=1}^{n - 1}\sigma^4n^4\sigma^{n/2} \leq n^5\sigma^{n/2 + 4}.
\]
\end{proof}

\section{A Sufficient Condition for Evasiveness}

Now that we have finished the proof of Theorem~\ref{thm:mainthm}, a natural question to ask is what patterns lie outside the scope of Theorem~\ref{thm:mainthm}? Rivest has given an example in \cite{rivest_worst-case_1977}, by showing that the pattern $1^n$ is non-evasive while $\gcd(\period(1^n)) = 1$ since every integer between 1 and $n$ is a period of $1^n$. In this section, we will use the algebraic method to develop a new sufficient condition for evasiveness. We will assume that the alphabet $\Sigma = \{0, 1\}$. We first introduce the notion of characteristic polynomial in Section~\ref{sec:subsec1} and then state our theorem in Section~\ref{sec:subsec2}. In Section~\ref{sec:subsec3}, we will show the relationship between a pattern's periods and its characteristic polynomial, which allows for a convenient way to calculate the polynomial.

\subsection{The KMP Automaton and the Transition Matrix}\label{sec:subsec1}
Following Rivest~\cite{rivest_worst-case_1977} we will make use of the finite state automaton constructed by the Knuth-Morris-Pratt algorithm. Let $p$ be a pattern string of length $m$, then the automaton constructed will have $m + 1$ states, where state $q_1$ is the initial state and state $q_{m + 1}$ is the only accepting state. The automaton reaches state $q_i$ if the previous $i-1$ characters is a prefix of $p$ where $i$ is the largest possible among such ones, and the pattern $p$ is not found already. The automation reaches state $q_{m+1}$ as soon as the pattern $p$ is found, and stays there forever. See Figure~\ref{fig:F1} for an example of the KMP automaton when the pattern $p = 1010$.

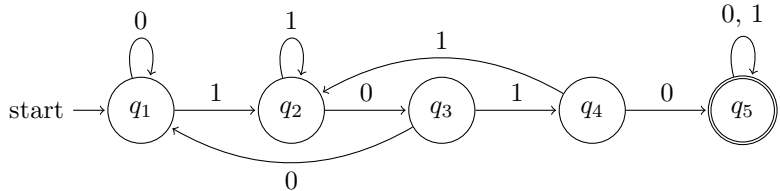
\begin{figure}[!h]
\centering
\begin{tikzpicture}[shorten >=1pt,node distance=2cm,on grid,auto]
   \node[state,initial] (q_1)   {$q_1$};
   \node[state] (q_2) [right=of q_1] {$q_2$};
   \node[state] (q_3) [right=of q_2] {$q_3$};
   \node[state] (q_4) [right=of q_3] {$q_4$};
   \node[state,accepting](q_5) [right=of q_4] {$q_5$};
    \path[->]
    (q_1) edge [loop above] node {0} ()
          edge  node  {1} (q_2)
    (q_2) edge  node  {0} (q_3)
          edge [loop above] node {1} ()
    (q_3) edge [bend left] node  {0} (q_1)
          edge  node  {1} (q_4)
    (q_4) edge  node  {0} (q_5)
          edge [bend right] node [above] {1} (q_2)
    (q_5) edge [loop above] node {0, 1} ();
\end{tikzpicture}
\caption{The finite state automaton for pattern $p = 1010$}
\label{fig:F1}
\end{figure}

Let $U_p(n, i)$ be the set of strings of length $n$ on which the automaton ends in state $q_i$. Let $g_p(n, i) := \sum_{s \in U_p(n, i)}x^{wt(s)}$, where $wt(s)$ is the number of 1's in $s$. The following lemma is used in \cite{rivest_generalization_1975} to show evasiveness of boolean functions.

\begin{lemma}[\cite{rivest_generalization_1975}]\label{lem:lem21}
If $D_p(n) \leq n - l$ for some integer $1 \leq l \leq n$, then $(x+1)^l$ divides $g_p(n, m + 1)$.
\end{lemma}

A useful consequence of this lemma is the following corollary.

\begin{corollary}
If there exists $ N_0 \in \mathbb{N}$ such that $g_p(n, m + 1) \not\equiv 0 \mod (x + 1)$ for all $n > N_0$, then $D_p(n)=n$, i.e. $p$ is evasive.
\end{corollary}

By Lemma~\ref{lem:lem21}, we are only interested in the value of $g_p(n, m + 1)$ modulo $x + 1$. Note that we always have
\[
g_p(n + 1, m + 1) = (x + 1)g_p(n, m + 1) + y\cdot g_p(n, m),
\]
where $y$ equals $1$ or $x$ depending on the last bit of the pattern string. Taking modulus of $(x+1)$ on both sides, we obtain
\[
g_p(n + 1, m + 1) \equiv y \cdot g_p(n, m) \mod (x+1).
\]
Since $y \equiv \pm 1 \mod (x+1)$ (with the sign determined by the last bit of the pattern string), we obtain the following lemma.

\begin{lemma}\label{lem:lem23}
$g_p(n + 1, m + 1) \equiv 0 \mod (x+1)$ if and only if $g_p(n, m) \equiv 0 \mod (x+1)$. Moreover, if there exists $ N_0 \in \mathbb{N}$ such that $g_p(n, m) \not\equiv 0 \mod (x + 1)$ for all $n > N_0$, then $p$ is evasive.
\end{lemma}

Now we define the transition matrix. Given a pattern string $p$ of length $m$, we can express $g_p(n + 1, 1), \ldots, g_p(n + 1, m)$ in terms of $g_p(n, 1), \ldots, g_p(n, m)$. For example, when $p = 1010$, according to the automata in Figure~\ref{fig:F1}, we can write
\[
\left(
  \begin{array}{c}
    g_p(n + 1, 1) \\
    g_p(n + 1, 2) \\
    g_p(n + 1, 3) \\
    g_p(n + 1, 4) \\
  \end{array}
\right)
=
\left(
  \begin{array}{cccc}
    1 & 0 & 1 & 0 \\
    x & x & 0 & x \\
    0 & 1 & 0 & 0 \\
    0 & 0 & x & 0 \\
  \end{array}
\right)
\left(
  \begin{array}{c}
    g_p(n, 1) \\
    g_p(n, 2) \\
    g_p(n, 3) \\
    g_p(n, 4) \\
  \end{array}
\right)
\]

Since we are only interested in these values modulo $x+1$, we may plug in $x = -1$ into all these terms. We denote $\overline{g}_p(n, i)$ to be the value obtained by plugging $x = -1$ into $g_p(n, i)$. In the previous example where $p = 1010$, we will obtain
\[
\left(
  \begin{array}{c}
    \overline{g}_p(n + 1, 1) \\
    \overline{g}_p(n + 1, 2) \\
    \overline{g}_p(n + 1, 3) \\
    \overline{g}_p(n + 1, 4) \\
  \end{array}
\right)
=
\left(
  \begin{array}{cccc}
    1 & 0 & 1 & 0 \\
    -1 & -1 & 0 & -1 \\
    0 & 1 & 0 & 0 \\
    0 & 0 & -1 & 0 \\
  \end{array}
\right)
\left(
  \begin{array}{c}
    \overline{g}_p(n, 1) \\
    \overline{g}_p(n, 2) \\
    \overline{g}_p(n, 3) \\
    \overline{g}_p(n, 4) \\
  \end{array}
\right)
\]

We call the matrix on the right hand side of the above equation the \emph{transition matrix} of the pattern string $p = 1010$. In general, given a pattern string $p$, we write down the recurrence relation of $g_p(n, i)$ in the matrix form and plug in $x=-1$, and the resulting matrix will be our transition matrix. Let $T_p$ denote the transition matrix for pattern $p$. 

We will see later that the eigenvalues of $T_p$ are of great use to us. We establish the following lemma using the characteristic polynomial of $T_p$. In the remaining part of this paper we will refer to the characteristic polynomial of $T_p$ as characteristic polynomial of the pattern $p$.
\begin{lemma}\label{lem:lemch}
Let $p$ be a pattern of length $m$. Let $P(\lambda) = \lambda^m + c_{m-1}\lambda^{m-1} + \cdots + c_0$ be the characteristic polynomial of $p$, then we have the recurrence relation
\begin{equation}\label{eq:eqch}
\overline{g}_p(n + m, m) + c_{m - 1}\overline{g}_p(n + m - 1, m) + \cdots + c_0\overline{g}_p(n, m) = 0.
\end{equation}
\end{lemma}
\begin{proof}
By the Cayley-Hamilton theorem (see Theorem 5.2.3 in \cite{artin_algebra_2011}), we have
\[
T_p^m + c_{m - 1}T_p^{m - 1} + \cdots + c_0I = 0,
\]
where $I$ is the identity matrix. Right multiply both sides by column vector
\[
\overline{g}_p(n) = (\overline{g}_p(n, 1),\, \overline{g}_p(n, 2),\, \ldots,\, \overline{g}_p(n, m))^t,
\]
we obtain
\[
\overline{g}_p(n + m) + c_{m - 1}\overline{g}_p(n + m - 1) + \cdots + c_0\overline{g}_p(n) = 0,
\]
since $T_p\overline{g}_p(n) = \overline{g}_p(n + 1)$. Both sides of the equation above are $m$-dimensional column vectors, and we get the desired recurrence relation by looking at the last row of both vectors.
\end{proof}
\subsection{The Skolem Problem and Finite Zeroes}\label{sec:subsec2}
Lemma~\ref{lem:lemch} gives us a tool to get around $\overline{g}_p(n, 1), \ldots, \overline{g}_p(n, m-1)$ and focus only on $\overline{g}_p(n, m)$. Now we are faced with the following problem:
\begin{quote}
Let $\{u_n\}$ be a linear recurrent sequence. Does there exist $N_0$ such that $u_n$ is non-zero for all $n > N_0$?
\end{quote}
This problem is very similar to the \emph{Skolem's Problem}, which can be stated as follows:
\begin{quote}
Let $\{u_n\}$ be a linear recurrent sequence. Does there exist $n$ such that $u_n = 0$?
\end{quote}

For a detailed survey of the Skolem's problem, readers are referred to \cite{halava_skolem_2005}. We will use the following result from \cite{halava_skolem_2005}, which partially solved our problem.

\begin{lemma}[\cite{halava_skolem_2005}]\label{lem:skolemlem}
Assume sequence $\{u_n\}_{n = 1}^\infty$ satisfies
\[
u_n = a_{m - 1}u_{n-1} + \cdots + a_1u_{n - m + 1} + a_0u_{n-m},
\]
where $a_0, a_1, \ldots, a_{m-1}$ are fixed integers. Also assume that $p(\lambda) = \lambda^m - a_{m-1}\lambda^{m-1} - \cdots - a_1\lambda - a_0$ has the decomposition
\[
p(\lambda) = (\lambda - \lambda_1)^{m_1}(\lambda - \lambda_2)^{m_2}\cdots(\lambda - \lambda_r)^{m_r},
\]
where $\lambda_1, \ldots, \lambda_r \in \mathbb{C}$ are distinct roots of $p(\lambda)$ and $|\lambda_1| \geq |\lambda_2| \geq \cdots \geq |\lambda_r|$. Then there exists $N_0 \in \mathbb{N}$ such that $u_n$ is non-zero for all $n > N_0$ if one of the following cases holds:
\begin{enumerate}
  \item $|\lambda_1| > |\lambda_2|$.
  \item $|\lambda_1| = |\lambda_2| > |\lambda_3|,\ \lambda_1 = \overline{\lambda_2}$.
  \item $|\lambda_1| = |\lambda_2| = |\lambda_3| > |\lambda_4|,\ \lambda_1 \in \mathbb{R},\ \lambda_2 = \overline{\lambda_3}$.
\end{enumerate}
\end{lemma}

The proof of this lemma can be found in Proposition 4.1 in \cite{halava_skolem_2005}. Note that our statement is a little bit different. In \cite{halava_skolem_2005} it is proved that the Skolem's problem is decidable in these cases, by showing that there exists an algorithmically computable constant $N_0$ such that $u_n \neq 0$ for all $n \geq N_0$, and therefore an algorithm for deciding the Skolem's problem only needs to check whether there are zeroes below the bound $N_0$.

Using this lemma, we can show the evasiveness of some pattern strings. As an example, we prove the following proposition.
\begin{prop}
The pattern $p = 10^k1$ is evasive when $k > 0$.
\end{prop}
\begin{proof}
To begin with, we calculate its characteristic polynomial, which is
\[
p(\lambda) = \det (\lambda I - T_p) = \lambda^{k+2} - \lambda + 1.
\]
We note that this is also the characteristic polynomial for the recurrence of $\overline{g}(n, k + 2)$ (here $|p| = k + 2$). Now assume $z = re^{i\theta}$ is a root of $p(\lambda) = 0$. Then we have $|z^{k + 2}| = |z - 1|$,
which implies
\[
r^{k + 2} = \sqrt{r^2 + 1 - 2r\cos{\theta}}.
\]
This shows that for every $r$ the value of $\cos \theta$ is determined, and therefore there can be at most 2 choices of $\theta$. Thus, the pattern $p$ either satisfies condition 1 or condition 2 in Lemma~\ref{lem:skolemlem}. We conclude that $p$ is evasive.
\end{proof}
\begin{remark}
The evasiveness of pattern $p = 10^k1$ is not covered by Tuza's Theorem. Though $p \in BE(b)$ where $b = 1$, the pattern $10^k110^k1$ has a substring $p' = 0^k11$ which differs from $p$ in only two positions, and thus violates the condition (b) in Theorem~\ref{thm:tuzathm}.
\end{remark}

\subsection{The Characteristic Polynomial and Periods}\label{sec:subsec3}
Writing down the characteristic polynomial through the transition matrix can sometimes be inefficient. Here we develop a faster way to calculate a pattern's characteristic polynomial and show some interesting connection to the periodic behavior of the pattern.

We give the following formula for a pattern's characteristic polynomial in terms of the pattern's periods.

\begin{theorem}\label{thm:polythm}
Let $p$ be a pattern of length $m$. Let $P(\lambda)$ be the characteristic polynomial of $p$, then we have
$
P(\lambda) = \lambda^m + c_{m-1}\lambda^{m-1} + \cdots + c_1\lambda + c_0,
$
where for $1 \leq k \leq m$,
 \[
    c_{m-k} = \left\{\begin{array}{cl}
        (-1)^{wt(p[1..k])}, & \text{if } k \text{ is a period of } p, \\
        0, & \text{otherwise.}
        \end{array} \right.
 \]
\end{theorem}

In proving the above theorem, the following two lemmas will be useful.
\begin{lemma}\label{lem:polyl1}
Let $p$ be a pattern of length $m$. Assume that state $q_m$ of the KMP automaton for $p$ has a transition back to state $q_{m - k + 1}$ where $k \leq m$.
\begin{itemize}
  \item If $k = m$, then all patterns of $p[1..m-1]$ are preserved. That is to say, if $p[1..m-1]$ is $l$-periodic for some $l$, then $p[1..m]$ is also $l$-periodic.
  \item If $k < m$, then $k$ is the smallest of the periods of $p[1..m-1]$ which are destroyed. That is to say, $p[1..m-1]$ is $k$-periodic while $p[1..m]$ is not $k$-periodic, and furthermore, if $p[1..m-1]$ is $l$-periodic for some $l < k$, then $p[1..m]$ is also $l$-periodic.
\end{itemize}
\end{lemma}
\begin{proof}
We first examine the case where $k < m$. In this case, we have $p[1..m-k] = p[(k+1)..(m-1)]\overline{p[m]}$. This implies that $p[1..m-k-1] = p[(k+1)..(m-1)]$, and by our definition of periodicity, this means that $p[1..m-1]$ is $k$-periodic. However, since $p[m-k] = \overline{p[m]} \neq p[m]$, $p$ is not $k$-periodic. Notice that the KMP algorithm, when fails to match, will always find the largest $k$ such that $p[1..m-k] = p[(k+1)..(m-1)]\overline{p[m]}$, this proved the ``furthermore" part in the second bulletin.

For the case where $k = m$, we simply observe that if $p[1..m-1]$ is $l$-periodic for some $l \geq 1$ but $p[1..m]$ is not, the KMP automaton in state $q_m$ can transit to state $q_{m - l + 1}$ which is a better choice than state $q_1$.
\end{proof}

\begin{lemma}\label{lem:polyl2}
Let $p$ be a pattern of length $m$. Let $P_i(\lambda)$ be the characteristic polynomial of the pattern $p[1..i]$. Assume that state $q_m$ of the KMP automaton for $p$ has a transition back to state $q_{m - k + 1}$ where $1 \leq k \leq m$. Then
\[
P_m(\lambda) = \left\{\begin{array}{ll}
        \lambda P_{m-1}(\lambda) - (-1)^{wt(p[1..k])}P_{m - k}(\lambda), & \text{if } k < m, \\
        \lambda P_{m-1}(\lambda) + (-1)^{wt(p[1..m])}, & \text{if } k = m.
        \end{array} \right.
\]
\end{lemma}
\begin{proof}
Let's consider the transition matrix $T_p$ of $p$. Note that since in the KMP automaton state $q_i$ has no transition to state $q_{i + j}$ for $j \geq 2$, the entries below the subdiagonal are all zero in $T_p$. Also, since state $q_i$ always has a transition to state $q_{i + 1}$, the entries on the subdiagonal are all non-zero.

Now we calculate the characteristic polynomial of $p$, which is $\det (\lambda I - T_p)$. In order to calculate this determinant, we use expansion by minors on the $n$-th column. To simplify notation we let $A = \lambda I - T_p$, $a_{ij}$ be the entry of $A$ appearing in the $i$-th row and $j$-th column, and $A_{ij}$ be the $(m - 1) \times (m - 1)$ matrix obtained by removing the $i$-th row and $j$-th column from $A$. We note that there is a change in sign for the terms from $T_p$ in $A$.

Since every state has at most two transitions, each column has at most two non-zero entries. For the $m$-th column, one will be $\lambda$ on the diagonal, while the other is 1 or -1 on the $(m - k + 1)$-th row representing the transition from state $q_m$ back to state $q_{m - k + 1}$ (these two may in fact be in the same entry when $k=1$). Thus our expansion consists of two terms. The first term is $\lambda$ multiplied by the $\det A_{mm}$. By induction hypothesis, $\det A_{mm} = P_{m-1}(\lambda)$, and therefore the first term will be $\lambda P_{m-1}(\lambda)$. The second term, which is $(-1)^{(m - k + 1) + m}a_{(m-k+1)m}\det A_{(m-k+1)m}$, is slightly more cumbersome. We analyze it in two cases.

\begin{itemize}
  \item $k < m$. In this case, we have $m - k + 1 \geq 2$. The entry $a_{(m-k+1)m}$, which comes from the transition from state $q_m$ back to state $q{m - k + 1}$, can be computed by
  \[
    a_{(m-k+1)n} = (-1) \times (-1)^{wt(p[m-k])}
  \]
  The first (-1) factor comes from the change of sign from $T_p$ to $A$, as noted earlier, and the second term comes from the fact that $\overline{p[m]} = p[m-k]$. To compute $\det A_{(m-k+1)m}$, we partition $A_{(m-k+1)m}$ into blocks as follows.
  \[
    A_{(m-k+1)m} = \left(
               \begin{array}{cc}
                 X & Y \\
                 0 & Z \\
               \end{array}
             \right)
  \]
  Here $X$ is of size $(m-k) \times (m-k)$ and $Z$ is of size $(k - 1) \times (k - 1)$. We have shown earlier that the entries below the subdiagonal in $A$ are all zeroes, so the lower left block of our partition consists entirely of zeroes. Also, $Z$ will be a diagonal matrix, with entries of 1 and -1 in its diagonal. Therefore we have
  \[
  \det Z = (-1)^{k-1} \times (-1)^{wt(p[(m-k+1)..(m-1)])},
  \]
  where the first factor comes from the change in sign and the second factor from $T_p$.
  As for $X$, by induction hypothesis, we have $\det X = P_{m-k}(\lambda)$. So we have
  \[
   \det  A_{(m-k+1)m} = \det X \det Z = (-1)^{k - 1 + wt(p[(m-k+1)..(m-1)])}P_{m-k}(\lambda).
  \]
  So the second term in our expansion will be
  \begin{align*}
    & (-1)^{(m - k + 1) + m}a_{(m-k+1)m}\det A_{(m-k+1)m} \\
    = & (-1)^{m - k + 1 + m + 1 + wt(p[m-k]) + k - 1 + wt(p[(m-k + 1)..(n-1)])}P_{m-k}(\lambda) \\
    = & (-1)^{1 + wt(p[1..k])}P_{m-k}(\lambda)
  \end{align*}
  The last equality comes from the fact that $p[m-1]$ is $k$-periodic (by Lemma~\ref{lem:polyl1}) and $wt(p[(m-k)..(m-1]) = wt(p[1..k])$. Thus we have
  \[
    P_m(\lambda) = P_{m-1}(\lambda) - (-1)^{wt(p[1..k])}P_{m - k}(\lambda).
  \]
  \item $k = m$. In this case, $A_{(m-k+1)m} = A_{1m}$ will be a diagonal matrix, and we have
  \[
    \det A_{1m} = (-1)^{m-1} \times (-1)^{wt(p[1..m-1])}
  \]
   and
  \[
    a_{1m} = (-1) \times (-1)^{1 - wt(p[m])}.
  \]
  So the second term in our expansion will be
  \begin{align*}
    & (-1)^{1 + m}a_{1n}\det A_{1m} \\
    = & (-1)^{1 + m + 1 + 1 - wt(p[m]) + m - 1 + wt(p[1..m-1])} \\
    = & (-1)^{wt(p[1..m])}
  \end{align*}
  and therefore
  \[
  P_m(\lambda) = \lambda P_{m-1}(\lambda) + (-1)^{wt(p[1..m])}.
  \]

\end{itemize}
\end{proof}

\begin{proof}[Proof of Theorem~\ref{thm:polythm}]
We use mathematical induction in this proof. The basis case where $m=1$ is easy to verify. Now we assume that $m > 1$ and the theorem holds for patterns with length smaller than $n$. If $k = m$, then by Lemma~\ref{lem:polyl2} we have that
\[
  P_m(\lambda) = \lambda P_{m-1}(\lambda) + (-1)^{wt(p[1..m])}.
\]
By the first bulletin in Lemma~\ref{lem:polyl1}, all periods in $p[1..m-1]$ are preserved in this case, and by induction hypothesis, they are represented in $P_{m-1}(\lambda)$. We note that if $p$ is $l$-periodic for some $l < m$, then $p[1..m-1]$ must be $l$-periodic as well, so there will be a term $(-1)^{wt(p[1..l])}\lambda^{m - 1 - l}$ in $P_{m-1}(\lambda)$, and therefore a term $(-1)^{wt(p[1..l])}\lambda^{m - l}$ in $\lambda P_{m-1}(\lambda)$. Besides the periods of $p[1..m-1]$, $p$ is also $m$-periodic, and this is represented in the term $(-1)^{wt(p[1..m])}$.

Now we analyze the case where $k < m$. Again, we apply Lemma~\ref{lem:polyl2} and write
\[
P_m(\lambda) = \lambda P_{m-1}(\lambda) - (-1)^{wt(p[1..k])}P_{m - k}(\lambda).
\]
Let $l$ be an integer with $1 \leq l \leq m$. We want to show that $p$ is $l$-periodic if and only if there is a term $(-1)^{wt(p[1..l])}\lambda^{m-l}$ in $P_m(\lambda)$.

\begin{itemize}
  \item $l = m$. By definition, $p$ is $m$-periodic. The corresponding term is the constant term, which can only come from $-(-1)^{wt(p[1..k])}P_{m - k}(\lambda)$. By induction, the constant term in $P_{m-k}(\lambda)$ is $(-1)^{wt(p[1..m-k])}$. So the constant term of $P_n(\lambda)$ will be
      $(-1)^{1 + wt(p[1..k]) + wt(p[1..m-k])}$. By Lemma~\ref{lem:polyl1}, $p[1..m-1]$ is $k$-periodic and $p$ is not, which means $p[1..m-1]\overline{p[m]}$ must be $k$-periodic. This implies
      \[
      (-1)^{1 + wt(p[1..k])} = (-1)^{1 + wt(p[(m-k+1)..m-1]\overline{p[m]})} = (-1)^{wt(p[(m-k+1)..m])}.
      \]
      Thus the constant term will be
      \[
      (-1)^{1 + wt(p[1..k]) + wt(p[1..m-k])} = (-1)^{wt(p[1..m])}.
      \]

  \item $k < l < m$. This case is a little bit trickier. First of all, assume $p$ is $l$-periodic, and therefore $p[1..m-1]$ is also $l$-periodic. By induction hypothesis, this implies that there is a term $(-1)^{wt(p[1..l])}\lambda^{m-l}$ in $\lambda P_{m-1}(\lambda)$. We need to show that this term does not appear in $P_{m - k}(\lambda)$. Suppose it does, then we know that, since $P_{m - k}(\lambda)$ is of degree $m-k$, $m - k - ( m - l ) = l - k$ must be a period of $p[1..m-k]$. So we have $p[m] = p[m - l] = p[m - k]$, where the first equality is from the $l$ period of $p$ and the second is from the $l-k$ period of $p[1..m-k]$. However, by hypothesis, adding $p[m]$ will destroy the $k$ period of $p[1..m-1]$ and therefore $p[m] \neq p[m-k]$, which leads to a contradiction.

      On the other hand, suppose that $p$ is not $l$-periodic. In this case, since $p[1..m-1]$ is $k$-periodic, it is easy to see that $p[1..m-1]$ is $l$-periodic if and only if $p[1..m - k]$ is $(l-k)$-periodic. So the two terms of degree $m-l$, if appearing in $P_m(\lambda)$, must appear together. When they appear together they will have coefficient
      \[
      (-1)^{wt(p[1..l])} - (-1)^{wt(p[1..k])}(-1)^{wt(p[1..l-k])},
      \]
      which is 0, since $p[1..l-k] = p[k + 1..l]$.
  \item $l = k$. By Lemma~\ref{lem:polyl1}, $p$ is not $k$-periodic but $p[1..m-1]$ is. This implies that there will be a term of $(-1)^{wt(p[1..k])}\lambda^{n-k}$ in $\lambda P_{m-1}(\lambda)$. This term will be cancelled out by the leading term in $(-1)^{wt(p[1..k])}P_{m - k}(\lambda)$, which has the same coefficient.
  \item $1 \leq l < k$. In this case, if $p$ has a period $l$, then the corresponding term in $P_m(\lambda)$ will have degree $m - l > m - k$, so it can only come from $\lambda P_{m-1}(\lambda)$. By Lemma~\ref{lem:polyl1}, $p[1..m-1]$ is $l$-periodic as well so such a term exists. On the other hand, if $p$ is not $l$-periodic, then by Lemma~\ref{lem:polyl1} $p[1..m-1]$ is not $l$-periodic as well, so such a term does not exist.
\end{itemize}

\end{proof}

\section{Conclusions}
In this paper, we determined the decision tree complexity of string matching problem for almost every string, except for those string the adversary method fails to give a lower bound, whose fraction is negligible. 

The algebraic approach in Section 4 further proves that a few of these strings are evasive. One open problem is to resolve the remaining cases.

The characteristic polynomial of a pattern $p$, which we encountered in Section~\ref{sec:subsec3}, might be of independent interest itself. We have shown that this polynomial is related to the pattern's periodic behaviour, and it will be interesting to investigate whether other properties of strings can be related to it.

Another natural extension is to consider randomized algorithms. All algorithms proposed in this paper are deterministic, and randomized complexity is still widely open.

\bibliographystyle{plain}
\bibliography{string_matching}

\begin{thebibliography}{10}

\bibitem{artin_algebra_2011}
M.~Artin.
\newblock {\em Algebra}.
\newblock Pearson Prentice Hall, 2011.

\bibitem{berstel_combinatorics_2003}
J.~Berstel and J.~Karhum{\"a}ki.
\newblock Combinatorics on words - a tutorial.
\newblock {\em Bulletin EATCS}, pages 178--228, February 2003.

\bibitem{cormen_introduction_2001}
T.~H. Cormen, C.~E. Leiserson, R.~L. Rivest, and C.~Stein.
\newblock {\em Introduction {To} {Algorithms}}.
\newblock MIT Press, 2001.

\bibitem{galil_time-space-optimal_1983}
Z.~Galil and J.~Seiferas.
\newblock Time-space-optimal string matching.
\newblock {\em Journal of Computer and System Sciences}, 26(3):280--294, June
  1983.

\bibitem{halava_skolem_2005}
V.~Halava, T.~Harju, M.~Hirvensalo, and J.~Karhum{\"a}ki.
\newblock Skolem's problem - on the border between decidability and
  undecidability.
\newblock {\em TUCS Technical Reports 683}, 2005.

\bibitem{karp_efficient_1987}
R.~M. Karp and M.~O. Rabin.
\newblock Efficient randomized pattern-matching algorithms.
\newblock {\em IBM Journal of Research and Development}, 31(2):249--260, March
  1987.

\bibitem{knuth_fast_1977}
D.~Knuth, J.~Morris, and V.~Pratt.
\newblock Fast {Pattern} {Matching} in {Strings}.
\newblock {\em SIAM Journal on Computing}, 6(2):323--350, June 1977.

\bibitem{nielsen_note_1973}
P.~Nielsen.
\newblock A note on bifix-free sequences ({Corresp}.).
\newblock {\em IEEE Transactions on Information Theory}, 19(5):704--706,
  September 1973.

\bibitem{rivest_worst-case_1977}
R.~L. Rivest.
\newblock On the {Worst}-{Case} {Behavior} of {String}-{Searching}
  {Algorithms}.
\newblock {\em SIAM Journal on Computing}, 6(4):669--674, December 1977.

\bibitem{rivest_generalization_1975}
R.~L. Rivest and J.~Vuillemin.
\newblock A {Generalization} and {Proof} of the {Aanderaa}-{Rosenberg}
  {Conjecture}.
\newblock In {\em Proceedings of the {Seventh} {Annual} {ACM} {Symposium} on
  {Theory} of {Computing}}, {STOC} '75, pages 6--11, New York, NY, USA, 1975.
  ACM.

\bibitem{tuza_worst-case_1982}
Z.~Tuza.
\newblock Worst-case behavior of string-searching algorithms.
\newblock {\em Journal of Statistical Planning and Inference}, 6(1):99--103,
  January 1982.

\bibitem{yao_complexity_1979}
A.~Yao.
\newblock The {Complexity} of {Pattern} {Matching} for a {Random} {String}.
\newblock {\em SIAM Journal on Computing}, 8(3):368--387, August 1979.

\end{thebibliography}
\end{document}